\documentclass[11pt,a4paper]{article}

\usepackage[intlimits]{amsmath}
\usepackage{amsthm,amsfonts}

\usepackage{cite,url}

\renewcommand{\Tilde}{\widetilde}
\renewcommand{\Hat}{\widehat}

\newcommand{\RR}{\mathbb{R}}
\newcommand{\ZZ}{\mathbb{Z}}
\newcommand{\CC}{\mathbb{C}}
\newcommand{\NN}{\mathbb{N}}

\newcommand{\cO}{\mathcal{O}}

\newtheorem{theorem}{Theorem}
\newtheorem*{theorem*}{Theorem}
\newtheorem{prop}[theorem]{Proposition}
\newtheorem{lemma}[theorem]{Lemma}
\newtheorem{corol}[theorem]{Corollary}

\theoremstyle{remark}

\DeclareMathOperator{\supp}{supp}

\begin{document}

\title{\bf Strong coupling asymptotics for a singular Schr\"odinger operator with an interaction supported by an open arc}

\author{{\sc Pavel Exner} \dag\quad and \quad {\sc Konstantin Pankrashkin} \ddag\\[\bigskipamount]
\dag{} Department of Theoretical Physics, Nuclear Physics Institute\\ Czech Academy of
Sciences, 25068 \v Re\v z near Prague, Czechia\\[\smallskipamount]
Doppler Institute for Mathematical Physics\\ and Applied Mathematics\\
Czech Technical University, B\v rehov\'a 7, 11519 Prague, Czechia\\
E-mail: \url{exner@ujf.cas.cz}\\[\medskipamount]
\ddag{} Laboratoire de math\'ematiques d'Orsay, UMR 8628\\
Universit\'e Paris-Sud 11, B\^atiment 425, 91400 Orsay, France\\
E-mail: \url{konstantin.pankrashkin@math.u-psud.fr}
}

\date{}
\maketitle

\begin{abstract}
\noindent We consider a singular Schr\"odinger operator in $L^2(\RR^2)$ written formally as 
$-\Delta - \beta\delta(x-\gamma)$ where $\gamma$ is a $C^4$ smooth open arc in $\RR^2$ of length $L$ with regular ends. It is shown that the $j$th negative eigenvalue of this operator behaves in the strong-coupling limit, $\beta\to +\infty$, asymptotically as 
\[
E_j(\beta)=-\frac{\beta^2}{4} +\mu_j^D +\cO\Big(\dfrac{\log\beta}{\beta}\Big),
\]
where $\mu_j^D$ is the $j$th Dirichlet eigenvalue of the operator
\[
-\frac{d^2}{ds^2} -\frac{\kappa(s)^2}{4}\,
\]
on $L^2(0,L)$ with $\kappa(s)$ being the signed curvature of $\gamma$ at the point $s\in(0,L)$
and $s$ is the arc length parameter.

\bigskip

\noindent{\it
This is a preliminary version. The final version appeared in 
Communications in Partial Differential Equations, Volume 39, Issue 2, 2014, pages 193--212
(published by Taylor \& Francis)
}
\end{abstract}

\section{Introduction}

Singular Schr\"odinger operators with interactions supported by manifolds of a lower dimension have been a subject of investigation in numerous papers, particularly in the last decade. One motivation came from physics where operators formally written as 
\[
H_\beta:=-\Delta - \beta\delta(x-\gamma)
\]
with $\beta>0$, where $\gamma$ is a metric graph embedded in a Euclidean space, are used as models of `leaky quantum graphs' describing motion of particles confined to a graph in a way allowing quantum tunneling between different parts of~$\gamma$.
At the same time there is a mathematical motivation to study such operators because they exhibit nontrivial and interesting relations between spectral properties and the geometry of the interaction support.
In the informal language, the above operator is the Laplacian with
the boundary conditions on $\gamma$, $[\partial f]=\beta f$, where
$[\partial f]$ denotes the jump of the normal derivative of $f$ on $\gamma$;
the rigorous definition is given by the associated sesquilinear form \cite{BEKS}:
$H_\beta$ is the self-adjoint operator associated with the form
\[
h_\beta(f,f)=\iint_{\RR^2}|\nabla f|^2dx - \beta\int_\gamma |f(x)|^2\,dS
\]
defined on $H^1(\RR^2)$, and the above boundary conditions should be understood
in a certain weak sense. Note that the operator in question
can be viewed as a special type mixed problem, cf. e.g. \cite{agr,grubb}.

An overview of known results concerning leaky quantum graphs is given in \cite{E08} which also offers a number of open problems. Some of them concern the \emph{strong-coupling behavior} of such operators. For large $\beta$ one expects the eigenfunctions corresponding to eigenvalues at the bottom of the spectrum to be strongly concentrated around $\gamma$ which suggests the asymptotic spectral behaviour might be determined by a one-dimensional problem.

If $\gamma$ has a finite length $L>0$, it is easy to see that the essential spectrum of $H_\beta$ is $[0,+\infty)$ and
that it has finitely
many negative eigenvalues. Denote these eigenvalues by $E_1(\beta)\le E_2(\beta)\le \dots \le E_j(\beta)\le \dots$
taking into account their multiplicities. The following result is shown in \cite{EY}:
\begin{prop}\label{prop0}
Assume that $\gamma$ is a $C^4$ smooth loop, then for each fixed $j\in\NN$ for $\beta\to+\infty$
one has the asymptotics
\[
E_j(\beta)=-\dfrac{\beta^2}{4}+\mu_j +\cO \Big(\dfrac{\log\beta}{\beta}\Big),
\]
where $\mu_j$ is the $j$th eigenvalue of the one-dimensional operator
\begin{equation}
     \label{eq-ddd}
D:=-\dfrac{d^2}{ds^2}- \dfrac{\kappa(s)^2}{4}
\end{equation}
on $(0,L)$ with the periodic boundary conditions; here $s$ in the length parameter and
$\kappa(s)$ is the signed curvature along $\gamma$.
\end{prop}
A similar result is available
for the case 
when $\gamma$ is an infinite smooth curve without
endpoints which is asymptotically straight in a suitable
sense \cite{E08}. We remark that the smoothness requirement is
essential. For example, if $\gamma$ consists of two
halflines meeting at a nonzero angle, a simple scaling
argument in combination with the existence result of \cite{EI01}
implies that the lowest eigenvalue is $E_1(\beta) = -c\beta^2$
with some $c>1/4$. We also refer to the paper \cite{lp} for
the discussion of similar effects for Robin Laplacians in
regions with non-smooth boundaries.

One asks naturally how such an asymptotics could look like if the curve $\gamma$ has endpoints.
The technique employed in \cite{EY} (which is briefly presented in the next section)
only gives a bracketing:
\begin{equation}
         \label{eq-brak}
-\dfrac{\beta^2}{4} + \mu_j^N+\cO \Big(\dfrac{\log\beta}{\beta}\Big)
\le E_j(\beta)\le
 -\dfrac{\beta^2}{4} + \mu_j^D+\cO \Big(\dfrac{\log\beta}{\beta}\Big),
\quad \beta\to+\infty,
\end{equation}
where $\mu_j^N$ and $\mu_j^D$ are respectively the $j$th Neumann and Dirichlet eigenvalues
of the above differential expression $D$ on $(0,L)$; note that this estimate was sufficient in \cite{EY}
to obtain an asymptotics of the eigenvalue counting function. An attempt was made in~\cite{KV}
to study the difference $E_2(\beta)-E_1(\beta)$, but the estimate obtained for the strong coupling regime
was quite weak.
A conjecture was made in Sec.~7.12 of \cite{E08} that under proper regularity assumptions the analog
of Proposition~\ref{prop0} holds if one replaces $\mu_j$ by $\mu_j^D$.
The aim of the present paper is to prove this conjecture in the case when $\gamma$ is a $C^4$ smooth compact arc in $\RR^2$
with regular endpoints. A precise formulation of this result is stated in the next section and the rest of the paper
is devoted to the proof.
As in the case of a curve without endpoints we employ a bracketing argument imposing Dirichlet and Neumann condition at the boundary of a tubular neighborhood of $\gamma$. In the present case, however, we need a neighbourhood extending beyond the endpoints and we loose the separation
of variables employed in \cite{EY}. Instead we have to establish the decay of eigenfunctions away of $\gamma$ which is technically the main part of the proof.

\section{Main result and first steps of the proof}

Let $\gamma$ be an open $C^4$ arc in $\RR^2$ of length $L>0$ and with regular ends.
More precisely, let $\Gamma:[0,L]\ni  s\mapsto  \big(\Gamma_1(s),\Gamma_2(s)\big)\in\RR^2$ 
be an injective map satisfying at any point $\big|\Gamma'(s)\big|=1$.
We set $\gamma:=\Gamma\big((0,L)\big)$.
Furthermore, we pick an arbitrary $l_0>0$ and fix an  extension
of $\Gamma$ to an injective $C^4$ map from $[-l_0,L+l_0]$ to $\RR^2$
satisfying $|\Gamma'(s)|=1$ for all $s\in (-l_0,L+l_0)$.
Denote by $\kappa(s)$ the signed curvature at $\Gamma(s)$, i.e.
\[
\kappa(s):=\Gamma'_1(s)\Gamma''_2(s)-\Gamma''_1(s)\Gamma'_2(s).
\]
Our main result reads as follows:
\begin{theorem}
         \label{thm1}
For any fixed $j\in\NN$, the asymptotic expansion
\[
E_j(\beta)=-\dfrac{\beta^2}{4}+\mu_j^D +\cO\Big(\dfrac{\log\beta}{\beta}\Big),
\]
holds for strong coupling, $\beta\to+\infty$, where $\mu_j^D$ is the $j$\emph{th} Dirichlet eigenvalue of the 
one-dimensional Schr\"odinger operator
\[
-\dfrac{d^2}{ds^2}-\dfrac{\kappa(s)^2}{4}
\]
with the curvature-induced potential on $(0,L)$.
\end{theorem}

Let us add that the optimality of the remainder term in
Proposition \ref{prop0} and Theorem \ref{thm1} represents an open question; we
shall not discuss it in the present paper. We only note that
it could be better in particular situations, e.g., the paper
\cite{ET} contained the example in which $\gamma$ was a circle and
the symmetry allowed to separate variables; in that case the
remainder in Proposition \ref{prop0} can be replaced by $\cO(\beta^{-2})$.

Before describing the scheme of the proof of Theorem~\ref{thm1}
let us briefly recall the construction of~\cite{EY} which was used to obtain
Proposition~1; this will make the meaning of the main result
more clear. Assume that $\gamma$ is a loop and that
$\gamma=\Gamma\big([0,L)\big)$, where
$\Gamma:\RR\ni s\mapsto\big(\Gamma_1(s),\Gamma_2(s)\big)\in \RR^2$ is an $L$-periodic
$C^4$ function whose restriction onto $[0,L)$ is an injection
and such that $|\Gamma'(s)|=1$ for all $s$.
We put
\begin{equation}
     \label{eq-tan}
\tau(s):=\begin{pmatrix}
\Gamma'_1(s)\\
\Gamma'_2(s)
\end{pmatrix},
\quad
n(s):=\begin{pmatrix}
-\Gamma'_2(s)\\
\Gamma'_1(s)
\end{pmatrix};
\end{equation}
in other words $\tau(s)$ is a unit tangent vector and $n(s)$ is a unit normal vector to $\gamma$ at the point $\Gamma(s)$, by assumption both continuously depending on the arc-length parameter.
Introduce the signed curvature $\kappa$ as previously and denote $\kappa_+:=\|\kappa\|_\infty$.
It can be shown that for sufficiently small $a>0$ the map
\[
\RR/L\ZZ\times (-a,a)\ni (s,t)\mapsto\Phi(s,t)=\Gamma(s)+tn(s)\in \RR^2
\]
is a diffeomorphism between $\square_a:=\RR/L\ZZ\times (-a,a)$
and $\Omega(a):=\Phi(\square_a)$ and
the map
$U_a:L^2\big(\Omega(a)\big)\to L^2(\square_a)$ defined by
$U_a f(s,t)=\sqrt{|J\Phi(s,t)|}f\big(\Phi(s,t)\big)$
is unitary; here $J\Phi$ is the Jacobian of $\Phi$.
Using the simple max-min argument one can show that for any  $j$
the inequalities $\Lambda_j^N(a,\beta)\le E_j(\beta)\le \Lambda_j^D(a,\beta)$
hold, where $\Lambda_j^{N/D}(a,\beta)$ are the properly ordered eigenvalues
of the operators $H^{N/D}(a,\beta)$ acting in $L^2\big(\Omega(a)\big)$
and associated with the forms $h^{N/D}_{a,\beta}$ given
by the identical expressions
\[
h^{N/D}_{a,\beta}(f,f)=\iint_{\Omega(a)}|\nabla f|^2dx-\beta \int_\gamma |f|^2ds,
\]
$h^N_{a,\beta}$ being defined on $H^1\big(\Omega(a)\big)$
and $h^D_{a,\beta}$ on $H^1_0\big(\Omega(a)\big)$.
The main ingredient of \cite{EY} was to show that the eigenvalues
$\Lambda_j^{N/D}(a,\beta)$ can be estimated using operators with separated variables
if $a$ and $\beta$ are related in a special way and $\beta$ is large.
More precisely, by estimating the expressions $h^{N/D}_{a,\beta}(U_a f,U_a f)$
and applying the max-min principle
it was shown that for any $j$ one has $\Tilde \Lambda_j^-(a,\beta)\le \Lambda_j^N(a,\beta)$
and $\Lambda_j^D(\beta,a)\le \Tilde \Lambda_j^+(a,\beta)$, where $\Tilde \Lambda_j^{\pm}(a,\beta)$
are the properly ordered eigenvalues of operators $\Tilde H^{\pm}_{a,\beta}$
in $L^2(\square_a)$ with separated variables, i.e.
$\Tilde H^\pm_{a,\beta}=L^\pm_a\otimes 1 + 1\otimes T^\pm_{a,\beta}$,
where $L^\pm_a$ are Schr\"odinger operators in $L^2(\RR/L\ZZ)$,
\[
L^\pm_a:=-(1\mp a \kappa_+)^2 \dfrac{d^2}{ds^2}+V^\pm_a (s)
\text{ with } V^\pm_a(s)=-\dfrac{\kappa(s)^2}{4}+\cO(a),
\]
acting on the common domain $H^2(\RR/L\ZZ)$,
and $T^\pm_{a,\beta}$ some explicit second order differential operators in $L^2(-a,a)$.
Let $\mu^\pm_j(a)$ denote the properly ordered eigenvalues
of $L^\pm_a$. It was observed in \cite{EY} that by taking $a:=6\beta^{-1}\log\beta$,
one achieves, for any fixed $j\in\NN$, the simultaneous validity of the estimates $\mu^\pm_j(a)=\mu_j+\cO(\beta^{-1}\log\beta)$
and of some estimates for the eigenvalues of $T^\pm_{a,\beta}$;
in particular, $\inf\sigma(T^\pm_{a,\beta})=-\beta^2/4+\cO(\beta^{-1}\log\beta)$.
Putting these estimates together gives the assertion of Proposition~\ref{prop0}.

Now let us get back to the situation we are concerned with here, i.e. a curve $\gamma$ with endpoints.
For any $\alpha\in(0,l_0)$
denote
\[
\Omega(\alpha):=\big\{\Gamma(s) +t n(s):\,
(s,t)\in (0,L)\times (-\alpha,\alpha)\big\}
\]
and compare $H_\beta$ with two operators acting in $\Omega(a)$ as in the case of a loop.
By applying the previously described constructions one can easily show the inequalities
$\Tilde \Lambda^-_j(a,\beta)\le E_j(\beta)\le \Tilde \Lambda^+_j(a,\beta)$,
where $\Tilde \Lambda^{\pm}_j(a,\beta)$ are respectively the eigenvalues of operators $\Tilde H^\pm_{a,\beta}$
with separated variables,
$\Tilde H^{\pm}_{a,\beta}=\Tilde L^{\pm}_a\otimes 1 + 1\otimes T^\pm_{a,\beta}$,
where $\Tilde L^-_a$ acts in $L^2(0,L)$ and is essentially given by the same expression as $L^-_a$
but with the Neumann boundary conditions, while $\Tilde L^+_a$
acts in $L^2(0,L)$ and is given by the same expression as $L^+_a$
but with the Dirichlet boundary conditions. 
Denoting $\tilde\mu_j^\pm(a)$ the respective
eigenvalues, it is easy to see that the estimate $\tilde\mu_j^+(a)
-\tilde\mu_j^-(a) = \mathcal{O}(\beta^{-1}\log\beta)$ which played
crucial role in the above described argument cannot be obtained
due to the fact that the boundary conditions at the interval ends
are different and the difference is $\cO(1)$,
and setting $a:=6\beta^{-1}\log\beta$ gives only the two-side estimate~\eqref{eq-brak}.

Therefore, we need to proceed differently and to supply the preceding approach
with a more detailed analysis of the eigenfunctions of $H_\beta$
near the endpoints of $\gamma$.
We employ the same notation
\eqref{eq-tan} for $s\in(-l_0,L+l_0)$ and denote $K:=\|\kappa\|_{L^\infty(-l_0,L+l_0)}$.
In addition to the above domain $\Omega(\alpha)$, for
$\alpha\in(0,l_0)$ let us introduce the following subdomains in $\RR^2$:
\[
P(\alpha):=(-\alpha,L+\alpha)\times (-\alpha,\alpha),\quad
\Pi(\alpha):=\big\{\Gamma(s) +t n(s):\,
(s,t)\in P(\alpha)\big\}
\]
and the prolonged arc
$\gamma_\alpha:=\Gamma\big((-\alpha,L+\alpha)\big) \subset \Pi(\alpha)$.
Clearly, $\gamma\subset\gamma_\alpha$ for any $\alpha>0$. Furthermore, one can check as in \cite{EY} that there is $a_0 \in \Big(0, \dfrac{1}{2K}\Big)$ such that the map
\begin{equation} \label{curvilinear}
P(a)\ni(s,t)\mapsto\Phi(s,t)=\Gamma(s)+tn(s)\in \Pi(a)
\end{equation}
is a diffeomorphism for any fixed $a\in (0,a_0]$. With the
previous discussion in mind,
throughout the rest of the paper we will always use
\begin{equation} \label{abeta}
a=\dfrac{6\log\beta}{\beta}.
\end{equation}
Let us introduce the following sesquilinear forms:
\begin{align*}
h_{\beta,a}(f,f)&=\iint_{\Pi(a)} |\nabla f|^2dx -\beta \int_\gamma |f|^2 dS, & \quad f\in H^1_0\big(\Pi(a)\big),\\
\Tilde h_{\beta,a}(f,f)&=\iint_{\Omega(a)} |\nabla f|^2dx -\beta \int_\gamma |f|^2 dS, & \quad f\in H^1_0\big(\Omega(a)\big),\\
\Hat h_{\beta,a}(f,f)&=\iint_{\Pi(a)} |\nabla f|^2dx -\beta \int_{\gamma_a} |f|^2 dS, & \quad f\in H^1_0\big(\Pi(a)\big),
\end{align*}
and denote the associated self-adjoint operators, acting respectively in $L^2\big(\Pi(a)\big)$, $L^2\big(\Omega(a)\big)$ and $L^2\big(\Pi(a)\big)$, by $L_\beta$, $\Tilde L_\beta$, $\Hat L_\beta$. We consider their eigenvalues $\Lambda_j(\beta)$, $\Tilde \Lambda_j(\beta)$, $\Hat\Lambda_j(\beta)$ enumerated in the non-decreasing order taking their multiplicities into account; by the max-min principle we have
\begin{equation}
       \label{eq-ela}
E_j(\beta)\le \Lambda_j(\beta).
\end{equation}

The asymptotic behavior of the right-hand side can be found easily:
\begin{prop}
      \label{prop1}
For any fixed $j\in\NN$, for sufficiently large $\beta$ one has
\begin{equation}
   \label{estl}
\Lambda_j(\beta)=-\dfrac{\beta^2}{4}+\mu_j^D +\cO\Big(\dfrac{\log\beta}{\beta}\Big).
\end{equation}
\end{prop}
\begin{proof}
Due to the max-min principle for any $j\in\NN$ we have the inequality
$\Hat \Lambda_j(\beta)\le \Lambda_j(\beta)\le \Tilde\Lambda_j(\beta)$.
Furthermore, the asymptotics of the estimating eigenvalues $\Tilde \Lambda_j$ and $\Hat\Lambda_j$ can be obtained using the technique introduced in \cite{EY} as explained above, and this gives
\[
\Tilde \Lambda_j(\beta)=-\dfrac{\beta^2}{4}+\mu^D_j +\cO\Big(\dfrac{\log\beta}{\beta}\Big),
\quad
\Hat \Lambda_j(\beta)=-\dfrac{\beta^2}{4}+\mu^D_j(\beta) +\cO\Big(\dfrac{\log\beta}{\beta}\Big),
\]
where $\mu_j^D$ and $\mu_j^D(\beta)$ are the $j$th Dirichlet eigenvalues of the operators given by the differential expression
$D$ from \eqref{eq-ddd} on $(0,L)$ and $(-a,L+a)$, respectively; recall that $a$ depends on $\beta$ as in~\eqref{abeta}. As the Dirichlet eigenvalues are $C^1$ functions of the interval edges, see e.g. \cite{DH}, we have $\mu^D_j(\beta)=\mu^D_j +\cO(a)=\mu_j^D+\cO(\beta^{-1}\log\beta)$, which proves the result.
\end{proof}

Hence the claim of Theorem \ref{thm1} will be a consequence of the following asymptotic relation:
\begin{prop}
        \label{prop2}
For any fixed $j\in\NN$ one has $\Lambda_j(\beta)-E_j(\beta)=\cO(\beta^{-1}\log\beta)$
as the coupling parameter $\beta$ tends to $+\infty$.
\end{prop}
This is our main estimate and the rest of the paper will dedicated to the its proof.
First, we collect in section~\ref{sec-tec} some technical details
concerning the analysis is small neighborhoods of $\gamma$.
In section~\ref{sec-ef} we use an integral representation of the eigenfunctions
to show suitable decay properties away of $\gamma$; this is the key point of the argument.
In section~\ref{sec-co} we introduce and analyze a family of cut-off eigenfunctions,
which are used in section \ref{sec-mm} as test functions for the max-min principle,
and this completes the proof.

We remark that the remainder estimate in Proposition~\ref{prop2} can be improved in various directions, in particular, by a suitable
inspection of the proof of~Lemma~\ref{lem2} below. Such improvements will not affect the remainder estimate
in Theorem~\ref{thm1}, however, as a larger error term comes already from \eqref{estl}, therefore, we did not try to
optimize in this direction.

\section{Technical estimates}\label{sec-tec}

We denote by $d(x,\gamma)$ the distance between a point $x\in\RR^2$ and the arc $\gamma$. In the present section we collect
some expressions for $d(x,\gamma)$ as $x\in \Pi(a)$ which we need in the following. 
Recall first the Frenet formul{\ae}
\begin{equation}
      \label{eq-tn}
\tau'(s)=\kappa(s) n(s), \quad n'(s)=-\kappa(s)\tau(s).
\end{equation}
In particular, for all $(s,t),(s',t')\in P(a_0)$
one has the representations
\begin{align}
      \label{eq-fg}
\Gamma(s')&=\Gamma(s)+(s'-s)\tau(s)+(s'-s)^2\rho_1(s',s),\\
      \label{eq-fn}
n(s')&=n(s) -(s'-s)\kappa(s)\tau(s) +(s'-s)^2\rho_2(s',s),\\
      \label{eq-ft}
\tau(s')&=\tau(s) +(s'-s)\kappa(s)n(s) +(s'-s)^2\rho_3(s',s),\\
\Phi(s',t')&=\Phi(s,t)+(s'-s)\big(1-t'\kappa(s)\big)\tau(s)+(t'-t)n(s) \nonumber\\
&\qquad
+(s'-s)^2\Big(
\rho_1(s',s)+t'\rho_2(s',s)\nonumber
\Big)\\
&=\Phi(s,t)+(s'-s)\big(1-t'\kappa(s)\big)\tau(s)+(t'-t)n(s) \nonumber\\
&\qquad + (s'-s)^2 \rho_4(s,t,s',t'),
         \label{eq-tay}
\end{align}
with
$\rho_1,\rho_2,\rho_3\in L^\infty\big((-a_0,L+a_0)^2\big)$ and
$\rho_4\in L^\infty\big(P(a_0)^2\big)$.
It is straightforward to check that for any $\alpha\in(0,a_0)$
 there are $C_1,C_2>0$ such that
\begin{equation}
        \label{eq-lipf}
C_1 \big((s-s')^2+(t-t')^2\big)\le
\big|\Phi(s,t)-\Phi(s',t')\big|^2\le
C_2 \big((s-s')^2+(t-t')^2\big)
\end{equation}
holds for all $(s,t),(s',t')\in P(\alpha)$.

The following lemma is rather obvious from the geometric point of view
and its proof is omitted as this is a routine based on the Frenet formul{\ae}.

\begin{lemma}\label{lemp3}
There exists $\alpha\in(0,a_0)$ such that
\[
d\big(\Phi(s,t),\gamma\big)= \begin{cases}
|t|, & \text{for } (s,t)\in(0,L)\times (-\alpha,\alpha),\\
\big|\Phi(s,t)-\Phi(0,0)\big|,& \text{for }  (s,t)\in(-\alpha,0)\times (-\alpha,\alpha),\\
\big|\Phi(s,t)-\Phi(L,0)\big|,& \text{for }  (s,t)\in(L,L+\alpha)\times (-\alpha,\alpha).
\end{cases}
\]
\end{lemma}
As a direct corollary we obtain
\begin{lemma}\label{lemp4}
For $s<0$ we have in the limit $(s,t)\to 0$ the relation
\begin{equation}
      \label{eq-fst}
d\big(\Phi(s,t),\gamma\big)=\sqrt{s^2+t^2} +\cO(s^2+t^2).
\end{equation}
Similarly, for $s>L$ and $(s,t)\to (L,0)$ we have
\[
d\big(\Phi(s,t),\gamma\big)=\sqrt{(s-L)^2+t^2} +\cO\big((s-L)^2+t^2\big).
\]
\end{lemma}

Applying Lemmata \ref{lemp3} and \ref{lemp4} to the boundary of $\Pi(a)$ we obtain
\begin{corol}\label{corol1}
There are $\alpha_0\in (0,a_0)$ and $C>0$ such that
for all $\alpha\in(0,\alpha_0)$ and $x\in \partial\Pi(\alpha)$
we have $d(x,\gamma)\ge \alpha-C\alpha^2$.
\end{corol}

For a fixed $b>0$ we introduce the set
\[
W(b):=\big\{x\in\RR^2:d(x,\gamma)<b\big\}.
\]
and derive an integral estimate on the complement of such a neighborhood:

\begin{lemma}
        \label{lem3}
Let $k,c>0$. In the limit $\beta\to+\infty$ we have
\[
\iint_{\RR^2\setminus W\big(\frac{k \log\beta-c}{\beta}\big)} e^{-(\beta-\log\beta)d(x,\gamma)}\,dx= \cO\Big(\dfrac{1}{\beta^{k+1}}\Big).
\]
\end{lemma}
\begin{proof}
During the demonstration we denote by $C_j$ various fixed positive numbers. Pick $p\in (0,1)$ with 
$p> \sqrt{\dfrac{k-1}{k}}$. Then by Lemmata \ref{lemp3} and \ref{lemp4}
one can find $\alpha>0$ such that
\begin{itemize}
\item $d\big(\Phi(s,t),\gamma\big)=|t|$ holds for all $s\in(0,L)$ and $t\in(-\alpha,\alpha)$,
\item $p\sqrt{s^2+t^2}\le d\big(\Phi(s,t),\gamma\big)\le p^{-1} \sqrt{s^2+t^2}$
holds for all $s\in(-\alpha,0)$ and $t\in(-\alpha,\alpha)$, and similarly,
\item $p \sqrt{(s-L)^2+t^2}\le d\big(\Phi(s,t),\gamma\big)\le p^{-1} \sqrt{(s-L)^2+t^2}$
holds for all $s\in(L,L+\alpha)$ and $t\in(-\alpha,\alpha)$.
\end{itemize}
One can represent the integration domain as follows:
\begin{align*}
\RR^2\setminus W\Big(\frac{k \log\beta-c}{\beta}\Big)&=
\Big[
W(\alpha)\setminus W\Big(\frac{k \log\beta-c}{\beta}\Big)
\Big]\\
&\quad
\cup
\Big[
W(2L)\setminus W(\alpha)
\Big]\cup
\big[\RR^2\setminus W(2L)\big].
\end{align*}
Let us estimate the contribution to the integral from each of these three components. Using the diffeomorphism $\Phi$ one easily reduces the integration on $W(\alpha)\setminus W\big(\frac{k \log\beta}{\beta}\big)$ to the integration on two rectangles and two half-discs: this yields the estimate
\begin{multline*}
\iint_{W(\alpha)\setminus W\big(\frac{k \log\beta-c}{\beta}\big)}e^{-(\beta-\log\beta)d(x,\gamma)}\,dx\\
\le C_1 \iint_{p\frac{k \log\beta-c}{\beta}\le|x|\le p^{-1}\alpha} e^{-p(\beta-\log\beta)|x|}\,dx
+C_2 \int_0^L \int_{\frac{k \log\beta-c}{\beta}}^{\alpha} e^{-(\beta-\log\beta)t}  \,dt \,ds\\
\le
C_3 \int_{p\frac{k \log\beta-c}{\beta}}^{p^{-1}\alpha} r e^{-p(\beta-\log\beta)r}\,dr
+
C_4 \int_{\frac{k \log\beta-c}{\beta}}^{\alpha} e^{-(\beta-\log\beta)t}  \,dt,
\end{multline*}
and the direct computation of the two integrals on the right-hand side gives
\[
\iint_{W(\alpha)\setminus W\big(\frac{k \log\beta-c}{\beta}\big)}
e^{-(\beta-\log\beta)d(x,\gamma)}\,dx=\cO\Big(\dfrac{1}{\beta^{k+1}}\Big).
\]

Furthermore, the measure of the second component, $W(2L)\setminus W(\alpha)$, certainly does not exceed $9\pi L^2$, while for all $x$ in this domain the integrated function is majorized by $e^{-(\beta-\log\beta) d(x,\gamma)}\le \beta^{\alpha} e^{-\alpha\beta}$, which gives
\[
\iint_{W(2L)\setminus W(\alpha)}
e^{-(\beta-\log\beta)d(x,\gamma)}dx
\le 9\pi L^2 \beta^{\alpha} e^{-\alpha\beta}=\cO(e^{-\alpha\beta/2}).
\]

Finally, to estimate the integral over the the complement of  $W(2L)$ let us pick a point $x_0\in\gamma$, then
for any $x\notin W(2L)$ one has
\[
d(x,\gamma)\ge |x-x_0|-L\ge |x-x_0|-\dfrac{|x-x_0|}{2}=\dfrac{|x-x_0|}{2}.
\]
Hence we have
\[
\iint_{\RR^2\setminus W(2L)} e^{-(\beta-\ln\beta)d(x,\gamma)}\,dx\le
\iint_{|x-x_0|>2L} e^{-(\beta-\ln\beta)|x-x_0|/2}\,dx
=2\pi \int_{2L}^{\infty} r e^{-(\beta-\ln\beta)r/2}\,dr
=\cO(e^{-L\beta/2}),
\]
and summing up the three terms one obtains the sought result.
\end{proof}

\section{Eigenfunctions estimates}\label{sec-ef}

To make some of the subsequent expressions more readable
we will use the following simplified form of the estimate \eqref{eq-brak}.
\begin{lemma}
            \label{lem1}
For any fixed $j\in\NN$ one has, as  $\beta\to+\infty$,

\[
 \dfrac{\beta -\ln\beta}{2}\le \sqrt{-E_j(\beta)}\le \dfrac{\beta +\ln\beta}{2}.
\]
\end{lemma}

Now let $u_{j,\beta}$ denote an $L^2$-normalized eigenfunction of $H_\beta$ corresponding to the eigenvalue $E_j(\beta)$, $j\in\NN$. 
Using an abstract analog of the boundary integral method, 
see e.g. Corollary~2.3 in~\cite{BEKS}, one can represent it as
\begin{equation}
               \label{eq-ffy}
u_{j,\beta}(x)=\int_\gamma G_0(x,y; E) F_{j,\beta}(y)\,dS_y,
\end{equation}
where $F_{j,\beta}\in L^2(\gamma)$ is an appropriate solution to the integral equation
\begin{equation}
               \label{eq-inteq}
\int_\gamma G_0\big(x,y; E_j(\beta)\big) F_{j,\beta}(y)dS_y = \dfrac{1}{\beta}\, F_{j,\beta}(x),\quad x\in\gamma,
\end{equation}
and $G_0$ is the Green function of the two-dimensional free Laplacian given explicitly by
\[
G_0(x,y;z)=\dfrac{1}{2\pi}K_0(\sqrt{-z}|x-y|).
\]
here and in the following $K_\nu$ denotes the modified Bessel function of order $\nu$, see \cite[Section 9.6]{AS};
alternatively, the integral representation can be obtained from the corresponding Krein's formula~\cite{pos}.
The following estimate will be of crucial importance for our result.
\begin{lemma}
      \label{lem2} $\|F_{j,\beta}\|_{L^2(\gamma)}=\cO(\beta^2\sqrt{\log\beta})$ holds as $\beta\to+\infty$.
\end{lemma}
\begin{proof}
Throughout the proof again $C_j$ will denote various positive constants. To avoid using cumbersome notation we identify the function $F_{j,\beta}(\cdot)$ with $F_{j,\beta}\big (\Phi(\cdot,0)\big)\equiv F_{j,\beta}\big(\Gamma(\cdot)\big)$
and write simply $E$ instead of $E_j(\beta)$.

We will employ the following well-know relation \cite[Eqs. 9.7.2 and 9.6.27]{AS}:
\begin{gather}
   \label{eq-kas}
K_\nu(w)=\sqrt{\dfrac{\pi}{2w}} e^{-w} \Big(1+o(1)\Big), \quad w\to +\infty, \quad \nu=0,1,\\
     \label{eq-k01}
K_0'=-K_1.
\end{gather}

According to \eqref{eq-ffy} and \eqref{eq-inteq}, one has
\begin{equation}
        \label{eq-uga}
u_{j,\beta}\big|_\gamma = \dfrac{1}{\beta}\,F_{j,\beta},
\end{equation}
and moreover, using \eqref{eq-ffy} and \eqref{eq-k01} we can write
\[
\nabla u_{j,\beta}(x)= \dfrac{1}{2\pi}\int_\gamma \dfrac{\sqrt{-E}(y-x)}{|x-y|} K_1\big(\sqrt{-E}|x-y|\big)\, F_{j,\beta}(y)\,dS_y.
\]
Another property to use \cite[Eqs. 9.6.10 and 9.6.11]{AS} is the representation
\begin{equation}
       \label{eq-ktt}
K_1(t)=\dfrac{1}{t} +M(t),
\quad M(t)= t g_1(t) \log t +g_2(t),
\end{equation}
where $g_1$ and $g_2$ are analytic functions. It yields
\[
\nabla u_{j,\beta}(x)= \dfrac{1}{2\pi}\int_\gamma \dfrac{(y-x)}{|x-y|^2} F_{j,\beta}(y) \,dS_y
+\dfrac{1}{2\pi}\int_\gamma \dfrac{\sqrt{-E}(y-x)}{|x-y|} M\big(\sqrt{-E}|y-x|\big)F_{j,\beta}(y) \,dS_y.
\]
Let us estimate the expression $n(s)\cdot\nabla u_{j,\beta}\big(\Phi(s,t)\big)$. In view of the representation  \eqref{eq-ktt} and the asymptotics \eqref{eq-kas}
we have a uniform bound $\big|M(w)\big|\le 2\pi C_1$ for all $w>0$, and therefore
\begin{multline}
           \label{eq-ue1}
\Big|n(s)\cdot\nabla u_{j,\beta}\big(\Phi(s,t)\big)\Big|
 \le
\Big|
\dfrac{1}{2\pi}
\int_0^L \dfrac{n(s)\cdot\big(\Phi(\sigma,0) -\Phi(s,t)\big)}{|\Phi(\sigma,0) -\Phi(s,t)|^2} F_{j,\beta}(\sigma) \,d\sigma
\Big|
+
C_1 \sqrt{-E}\|F_{j,\beta}\|_{L^1(\gamma)}.
\end{multline}
Furthermore, for large enough $\beta$ the inequalities \ref{eq-lipf} imply the estimate
\[
\dfrac{1}{2\pi|\Phi(\sigma,0) -\Phi(s,t)|^2}\le \dfrac{C_4}{(s-\sigma)^2 + t^2}.
\]
for  all $s,\sigma\in(0,L)$ and $t\in(-a,a)$, recall the assumption \eqref{abeta}. Next note that $\Phi(\sigma,0) -\Phi(s,t)=-t n(s) +\Gamma(\sigma)-\Gamma(s)$, hence
using \eqref{eq-fg} we get
\[
n(s)\cdot\big(\Phi(\sigma,0) -\Phi(s,t)\big)=- t+ (\sigma-s)^2 \rho(\sigma,s),
\]
where $\rho(\sigma,s)=n(s)\cdot \rho_1(\sigma,s)$ is uniformly bounded on $[0,L]\times[0,L]$. Consequently, there are $C_5, C_6>0$ such that
\[
\Big|
\dfrac{n(s)\cdot\big(\Phi(\sigma,0) -\Phi(s,t)\big)}{|\Phi(\sigma,0) -\Phi(s,t)|^2}\Big|
\le C_5 \dfrac{|t|+C_6(s-\sigma)^2}{(s-\sigma)^2 + t^2} \le
C_5 \dfrac{|t|}{(s-\sigma)^2 + t^2} + C_5C_6
\]
and
\begin{multline*}
\bigg|
\dfrac{1}{2\pi}
\int_0^L \dfrac{n(s)\cdot\big(\Phi(\sigma,0) -\Phi(s,t)\big)}{|\Phi(\sigma,0) -\Phi(s,t)|^2} F_{j,\beta}(\sigma) \,d\sigma
\bigg|
\\
\le
C_5 \int_0^L
 \dfrac{|t|}{(s-\sigma)^2 + t^2}
| F_{j,\beta}(\sigma)|d\sigma
+C_5 C_6 \|F_{j,\beta}\|_{L^1(\gamma)}.
\end{multline*}
Using the Cauchy-Schwarz inequality we obtain
\begin{multline}
\int_0^L
 \dfrac{|t|}{(s-\sigma)^2 + t^2}
| F_{j,\beta}(\sigma)|\,d\sigma
\le
\bigg(
\int_0^L
 \dfrac{|t|^2}{\big((s-\sigma)^2 + t^2\big)^2}
\,d\sigma
\bigg)^{1/2}
\|F_{j,\beta}\|_{L^2(\gamma)}
\\
\le
\bigg(
\int_\RR
 \dfrac{|t|^2}{\big((s-\sigma)^2 + t^2\big)^2}
\,d\sigma
\bigg)^{1/2}
\|F_{j,\beta}\|_{L^2(\gamma)}\\
=
|t|\bigg(
\int_\RR
 \dfrac{d\sigma}{(\sigma^2 + t^2)^2}
\bigg)^{1/2}
\|F_{j,\beta}\|_{L^2(\gamma)}\\
=
|t|^{-1/2}\bigg(\int_\RR
 \dfrac{d\xi}{(\xi^2 + 1)^2}
\bigg)^{1/2}
\|F_{j,\beta}\|_{L^2(\gamma)}=C_7 |t|^{-1/2} \|F_{j,\beta}\|_{L^2(\gamma)}.
\end{multline}
Putting everything together and using a rough estimate $E=\cO(\beta)$ from Lemma \ref{lem1}, we get the bound
\begin{equation}
         \label{eq-nabu}
\Big|n(s)\cdot\nabla u_{j,\beta}\big(\Phi(s,t)\big)\Big|
\le C_8\big(|t|^{-1/2} +\beta\big)\|F_{j,\beta}\|_{L^2(\gamma)}
\end{equation}
with some constant $C_8>0$. Next we denote
$\delta:=(\beta^2 \log\beta)^{-1}$,
and for $\beta$ large enough we construct a new function $v$ on $\Omega(\delta)$ by
\[
v_{j,\beta}\big(\Phi(s,t)\big):=u_{j,\beta}\big(\Phi(s,0)\big) ,\quad
(s,t)\in(0,L)\times(-\delta,\delta),
\]
for which the triangle inequality yields
\begin{equation}
        \label{eq-minu}
\|u\|_{L^2(\Omega(\delta))}\ge \|v_{j,\beta}\|_{L^2(\Omega(\delta))}-
\|u_{j,\beta}-v_{j,\beta}\|_{L^2(\Omega(\delta))}.
\end{equation}
Using \eqref{eq-uga}, one can write now the following estimates:
\begin{multline}
         \label{eq-vvj}
\|v_{j,\beta}\|^2_{L^2(\Omega(\delta))} \ge C_9
\int_{-\delta}^\delta \int_0^L \big|u_{j,\beta}\big(\Phi(s,0)\big)\big|^2 \,ds\, dt\\
= \dfrac{C_9}{\beta^2} \int_{-\delta}^\delta \int_0^L \big|F_{j,\beta}(s)\big|^2 \,ds\, dt
= \dfrac{C^2_{10} \delta}{\beta^2}  \|F_{j,\beta}\|^2_{L^2(\gamma)}
=\dfrac{C^2_{10}}{\beta^4\log\beta} \|F_{j,\beta}\|^2_{L^2(\gamma)}.
\end{multline}
On the other hand, the second term on the right-hand side of \eqref{eq-minu} satisfies
\begin{equation}
          \label{eq-uvu}
\|u_{j,\beta}-v_{j,\beta}\|^2_{L^2(\Omega(\delta))} \le C_{11}
 \int_0^L \int_{-\delta}^\delta \Big|u_{j,\beta}\big(\Phi(s,t)\big)-u_{j,\beta}\big(\Phi(s,0)\big)\Big|^2 \,dt\, ds.
\end{equation}
To estimate the integrated function, we employ the relation
\[
\dfrac{d}{dt} u_{j,\beta}\big(\Phi(s,t)\big) = n(s)\cdot \nabla u_{j,\beta}\big(\Phi(s,t)\big),
\]
which yields, through \eqref{eq-nabu}, the bound
\begin{multline*}
\Big|u_{j,\beta}\big(\Phi(s,t)\big)-u_{j,\beta}\big(\Phi(s,0)\big)\Big|
\\=\Big|
\int_0^t n(s)\cdot\nabla u_{j,\beta}\big(\Phi(s,\xi)\big)\,d\xi
\Big|
\le
\int_0^{|t|}
\Big|
n(s)\cdot\nabla u_{j,\beta}\big(\Phi(s,\xi)\big)
\Big|
\,d\xi
\\
\le
C_8 \int_0^{|t|} \big(|\xi|^{-1/2} +\beta\big) \,d\xi\
\|F_{j,\beta}\|_{L^2(\gamma)}
=
C_8 \big(2|t|^{1/2}+ |t|\beta\big)\cdot
\|F_{j,\beta}\|_{L^2(\gamma)},
\end{multline*}
and consequently,
\begin{multline}
     \label{eq-umv}
\|u_{j,\beta}-v_{j,\beta}\|^2_{L^2\big(\Omega(\delta)\big)} \le 8C_8 C_{11}
 \int_0^L \int_{-\delta}^\delta \Big( |t| + \beta^2 t^2\Big) \,dt\, ds\, \|F_{j,\beta}\|^2_{L^2(\gamma)}\\
\le  C_{12}(\delta^2 + \delta^3\beta^2)\ \|F_{j,\beta}\|^2_{L^2(\gamma)} \le \dfrac{C_{13}^2}{\beta^4 \log^2\beta}\|F_{j,\beta}\|^2_{L^2(\gamma)}.
\end{multline}
Substituting finally \eqref{eq-vvj} and \eqref{eq-umv} into \eqref{eq-minu} we obtain
\begin{multline*}
1= \|u_{j,\beta}\|_{L^2(\RR^2)}\ge \|u_{j,\beta}\|_{L^2\big(\Omega(\delta)\big)}\\
\ge
\Big(
\dfrac{C_{10}}{\beta^2\sqrt{\log\beta}}
-\dfrac{C_{13}}{\beta^2 \log\beta}
\Big)\|F_{j,\beta}\|_{L^2(\gamma)}
\ge
\dfrac{C_{14}}{\beta^2\sqrt{\log\beta}}\|F_{j,\beta}\|_{L^2(\gamma)},
\end{multline*}
which gives the sought result.
\end{proof}
We remark that a more involved choice of the constant $\delta$ during the proof would result in a better norm majoration
which would enter all the subsequent estimates. As already explained after Proposition~\ref{prop2}, we did not do this
as such an improvement would not help us to improve the resulting eigenvalue estimate
in Theorem~\ref{thm1}.
\begin{lemma}
      \label{lem5}
For any $k,c>0$ one can find a $D>0$ such that
\begin{gather}
       \label{eq-u1}
\big|u_{j,\beta}(x)\big|\le D\beta^2  \exp\Big(-\dfrac{(\beta-\log\beta)d(x,\gamma)}{2}\Big),\\
       \label{eq-u2}
\big|\nabla u_{j,\beta}(x)\big|\le D\beta^3 \exp\Big(-\dfrac{(\beta-\log\beta)d(x,\gamma)}{2}\Big)
\end{gather}
holds whenever $x\notin W\Big(\dfrac{k\log\beta-c}{\beta}\Big)$.
\end{lemma}
\begin{proof}
Recall that we have the integral representation \eqref{eq-ffy} for the eignefunction $u_{j,\beta}$, hence using Lemma \ref{lem2} and Cauchy-Schwarz inequality we infer that
\begin{multline}
         \label{eq-ukk}
|u_{j,\beta}(x)|\le \sup_{y\in\gamma }\Big|K_0\big(\sqrt{-E_j(\beta)}|x-y|\big)\Big|\cdot \|F_{j,\beta}\|_{L^1(\gamma)}\\
\le C_1\beta^2\sqrt{\log\beta}\, \sup_{y\in\gamma}\Big|K_0\big(\sqrt{-E_j(\beta)}|x-y|\big)\Big|.
\end{multline}
For $x\notin W\Big(\dfrac{k\log\beta-c}{\beta}\Big)$ and  $y\in\gamma$ we have, using Lemma \ref{lem1},
\[
\sqrt{-E_j(\beta)}|x-y|\ge \sqrt{-E_j(\beta)}\,d(x,\gamma)
\ge \dfrac{\beta-\log\beta}{2}\,
\dfrac{k\log \beta-c}{\beta}= \dfrac{k\log \beta}{2}+\cO(1)
\]
as $\beta\to +\infty$. For fixed $x,y$ the asymptotics \eqref{eq-kas} and Lemma~\ref{lem1} imply
\begin{multline*}
\Big|K_0\big(\sqrt{-E_j(\beta)}|x-y|\big)\Big|
\le C_3 \Big(\sqrt{-E_j(\beta)}|x-y|\Big)^{-1/2}
\exp\Big(-\dfrac{(\beta-\log\beta)|x-y|}{2}\Big)\\
\le C_3
\Big(\sqrt{-E_j(\beta)}d(x,\gamma)\Big)^{-1/2}
\exp\Big(-\dfrac{(\beta-\log\beta)d(x,\gamma)}{2}\Big)\\
\le \dfrac{C_4}{\sqrt{\log\beta}}\exp\Big(-\dfrac{(\beta-\log\beta)d(x,\gamma)}{2}\Big).
\end{multline*}
Combining this inequality with \eqref{eq-ukk} we obtain the bound \eqref{eq-u1}. To estimate $\nabla u_{j,\beta}$ we use \eqref{eq-k01} and write
\[
\nabla u_{j,\beta}(x)=
-\sqrt{-E_j(\beta)}\int_\gamma \nabla_x|x-y| K_1 \Big(\sqrt{-E_j(\beta)}\,|x-y|\Big) F_{j,\beta}(y)dS_y.
\]
It is now enough note that $\big|\nabla_x|x-y|\big|\le 1$ and that
$E_j(\beta)=\cO(\beta)$ by Lemma~\ref{lem1}, hence estimating the integral again with the help of \eqref{eq-kas}
we arrive at the bound \eqref{eq-u2}.
\end{proof}

\section{Cut-off functions}\label{sec-co}

In this section we introduce a family of cut-off functions that will be used in the following when we will apply the max-min principle in the last step of the argument. 

We choose a  function $C^\infty$ function $\psi:\RR\to [0,+\infty)$
such that
\[
\psi(s)=1 \;\text { for }\; s\ge 0
\;\text{ and }\;
\psi(s)=0 \;\text{ for }\; s\le -1.
\]
Next we consider the function $\rho_a:P(a)\to\RR$,
\[
\rho_a(s,t)=\min\big\{|a-t|,|a+t|, L+a-s, s+a\big\},
\]
in other words, $\rho_a(s,t)$ is the distance between the point $(s,t)\in P(a)$ and the boundary of the rectangle $P(a)$. We use it to introduce the function $R_\beta:\Pi(a)\to \RR$ by
\[
R_\beta(x)=\rho_a\big(\Phi^{(-1)}(x)\big).
\]
where $\Phi^{(-1)}(x)$ means the pre-image of the point $x\in \Pi(a)$ with respect to the map \eqref{curvilinear} and the parameters are related by \eqref{abeta}. Finally, for sufficiently large $\beta$ we introduce the function $g_\beta:\RR^2\to \RR$ by
\begin{equation} \label{mollif}
g_\beta(x)=\begin{cases}
\psi\Bigg( \dfrac{\log R_\beta(x)+\log \beta}{\log\log \beta}
\Bigg), & x\in \Pi(a),\\
0, & \text{otherwise}.
\end{cases}
\end{equation}
Note that $g_\beta$ belongs to $H^1(\RR^2)$ and has a compact support since $g(x)=0$ for all $x\notin\Pi(a)$. In addition, we have $g(x)=0$ for those $x\in\Pi(a)$ that can be represented as $x=\Phi(s,t)$ with $\rho_a(s,t)\le (\beta\log\beta)^{-1}$. On the other hand, $g(x)=1$ holds for $x\in \Pi(a)$ with $R_\beta(x)\ge \beta^{-1}$. In particular,
\begin{align*}
\supp \nabla g_\beta &\subset \Theta(\beta):=\Big\{
\Phi(s,t): (s,t)\in P(a),\quad \dfrac{1}{\beta \log\beta}\le \rho_a(s,t)\le \dfrac{1}{\beta}
\Big\},\\
\supp (1\!-\!g)&\subset V(\beta):=\Big\{\Phi(s,t): (s,t)\in P(a), \quad \rho_a(s,t)\le \dfrac{1}{\beta}\Big\}.
\end{align*}

\begin{lemma}
         \label{lem4}
In the limit $\beta\to+\infty$ one has
\[
\iint_{\Theta(\beta)} \big|\nabla g_\beta(x)\big| \,dx =\cO(1)\text{ and }
\iint_{\Theta(\beta)} \big|\nabla g_\beta(x)\big|^2 \,dx =\cO(\beta\log\beta).
\]
\end{lemma}
\begin{proof}
Let $D_{s,t}\Phi$ denote the Jacobian matrix value of the map $\Phi$ at $(s,t)$.
We have
\begin{align*}
\nabla g_\beta\big(\Phi(s,t)\big)&=\psi'\Bigg( \dfrac{\log R_\beta\big(\Phi(s,t)\big)+\log \beta}{\log\log \beta}\Bigg)
\dfrac{\nabla R_\beta \big(\Phi(s,t)\big)}{R_\beta\big(\Phi(s,t)\big)\log\log\beta} \\
&=\psi'\Bigg( \dfrac{\log \rho_a(s,t)+\log \beta}{\log\log \beta}\Bigg) \dfrac{1}{\rho_a(s,t)\log\log\beta} \nabla \rho_a(s,t) (D_{s,t}\Phi)^{-1}.
\end{align*}
We have $\big|\nabla\rho_a(s,t)\big|\le 1$ and $\big\|(D_{s,t}\Phi)^{-1}\big\|\le M$ for some $M>0$
and all $(s,t)\in P(a)$ if $\beta$ is sufficiently large.
Hence it holds
\[
\Big|
\nabla g_\beta\big(\Phi(s,t)\big)
\Big|\le \dfrac{C_1}{\rho_a(s,t) \log\log\beta}.
\]
with some $C_1>0$, and
\begin{equation}
     \label{eq-iii}
\iint_{\Theta(\beta)} \big|\nabla g_\beta(x)\big|^\nu \,dx
\le
\dfrac{C_1^\nu C_2}{(\log\log\beta)^\nu}
\iint_{\substack{ (s,t)\in P(a),\\ \frac{1}{\beta\log\beta}\le\rho_a(s,t)\le\frac{1}{\beta}}} \dfrac{ds\,dt}{\rho_a(s,t)^\nu},
\quad \nu=1,2.
\end{equation}
The integration domain can be decomposed naturally into four rectangles and eight triangles.
Using the obvious symmetries, we can rewrite it as
\begin{multline*}
I_\nu(\beta):=\iint_{\substack{(s,t)\in P(a)\\ \frac{1}{\beta\log\beta}<\rho_a(s,t)<\frac{1}{\beta}}} \dfrac{ds\, dt}{\rho_a(s,t)^\nu}
= 2 \int_{-a+\frac{1}{\beta}}^{L+a-\frac{1}{\beta}} \int_{\frac{1}{\beta\log\beta}}^{\frac{1}{\beta}} \frac{1 }{t^\nu}\, dt \, ds\\
+2 \int_{\frac{1}{\beta\log\beta}}^{\frac{1}{\beta}} \int_{-a+\frac{1}{\beta}}^{a-\frac{1}{\beta}} \dfrac{1}{s^\nu} \, dt \,ds
+8 \int_{\frac{1}{\beta\log\beta}}^{\frac{1}{\beta}} \int_{\frac{1}{\beta\log\beta}}^s
\dfrac{1}{t^\nu}\, dt \,ds\le C_3 \int_{\frac{1}{\beta\log\beta}}^{\frac{1}{\beta}} \frac{dt }{t^\nu}.
\end{multline*}
for $\nu=1,2$ and some $C_3>0$. Hence
\[
I_1(\beta)\le C_3 \int_{\frac{1}{\beta\log\beta}}^{\frac{1}{\beta}} \frac{dt }{t}
=C_3 \log\log\beta
\text{ and }
I_2(\nu)=C_3 \int_{\frac{1}{\beta\log\beta}}^{\frac{1}{\beta}} \frac{dt }{t^2}
=C_3(\beta\log\beta-\beta).
\]
Finally, by \eqref{eq-iii} we infer that
\[
\iint_{\Theta(\beta)} \big|\nabla g_\beta(x)\big| dx
=\dfrac{C_1 C_2  I_1(\beta)}{\log\log\beta}
\le
\dfrac{C_1 C_2 C_3\log\log\beta}{\log\log\beta}=\cO(1)
\]
and
\[
\iint_{\Theta(\beta)} \big|\nabla g_\beta(x)\big|^2 dx
= \dfrac{C_1^2 C_2 I_2(\beta)}{(\log\log\beta)^2}
\le\dfrac{C_1^2 C_2C_3 (\beta\log\beta-\beta)}{(\log\log\beta)^2}=\cO(\beta\log\beta).
\]
holds as $\beta\to +\infty$ which we have set out to prove.
\end{proof}

\begin{lemma}
          \label{lem6}
For sufficiently large $\beta$ there is a constant $D>0$ such that
\begin{gather}
       \label{eq-fest}
\big|u_{j,\beta}(x)\big|\le \dfrac{D}{\beta},\\
       \label{eq-fnab}
\big|\nabla u_{j,\beta}(x)\big|\le D
\end{gather}
holds for all $x\in V(\beta)$.
\end{lemma}
\begin{proof}
By Corollary \ref{corol1} there exists a $C_1>0$ such that
\[
d(x,\gamma)\ge a-C_1a^2 \text{ for all } x\in \partial\Pi(a),
\]
holds provided $\beta$ is sufficiently large.
On the other hand, for any $x=\Phi(s,t)\in V(\beta)$
one can find $(s',t')\in\partial P(a)$ with
\[
\rho_a\big(s,t\big)=\sqrt{(s-s')^2+(t-t')^2}\le \dfrac{1}{\beta}.
\]
As $\partial \Pi(a)=\Phi\big(\partial\Pi(a)\big)$, it follows from Eq.~\ref{eq-lipf} that for all $x\in V(\beta)$
\[
d\big(x,\partial\Pi(a)\big)\le
\big|\Phi(s,t)-\Phi(s',t')\big|
\le C_2 \sqrt{(s-s')^2+(t-t')^2}
\le \dfrac{C_2}{\beta} 
\]
holds with some $C_2>0$. Consequently, for sufficiently large $\beta$ we have
\[
V(\beta)\subset\RR^2\setminus W\Big(a-\dfrac{2C_2}{\beta}\Big)
=\RR^2\setminus W\Big(\dfrac{6\log\beta-2C_2}{\beta}\Big),
\]
and Lemma \ref{lem5} is applicable. For $x\in V(\beta)$ and large $\beta$ we can estimate
\[
\sqrt{-E_j(\beta)}d(x,\gamma)\ge \dfrac{\beta-\log\beta}{2}
\Big(\dfrac{6\log\beta}{\beta }-\dfrac{2C_2}{\beta}\Big)=
3\log\beta+\cO(1),
\]
by Lemma \ref{lem1}, hence applying \eqref{eq-u1} and \eqref{eq-u2} we get the sought bounds.
\end{proof}

\section{Using the max-min principle}\label{sec-mm}

Let us fix now an integer $N\ge 1$. Consider the first $N$ eigenvalues $E_j(\beta)$
and the associated \emph{orthonormal} eigenfunctions $u_{j,\beta}$ of $H_\beta$
and denote
\[
\varphi_{j,\beta}:=g_\beta u_{j,\beta},
\]
where $g_\beta$ is the function \eqref{mollif}. As $\supp g_\beta\subset\Pi(a)$, one has $\varphi_{j,\beta}\in H^1_0\big( \Pi(a)\big)$. Following the usual convention, we denote here and in the following by $\delta_{jl}$ the Kronecker delta symbol.

\begin{lemma}
        \label{lemort}
In the limit $\beta\to+\infty$ one has
\begin{equation}
       \label{eq-ort}
\langle \varphi_{j,\beta}, \varphi_{l,\beta}\rangle_{L^2(\Pi(a))}=\delta_{jl}+\cO(\beta^{-2}).
\end{equation}
\end{lemma}
\begin{proof}
Denote for brevity $S_\beta:=W\Big(\dfrac{5\log\beta}{\beta}\Big)$. In a way similar to the proof of Lemma \ref{lem6} one can show that for all sufficiently large $\beta$ we have $S_\beta\subset\Pi(a)$ and $g_\beta\big|_{S_\beta}=1$. Moreover, for $x\notin S_\beta$ one can estimate $u_{j,\beta}(x)$ with the help of Lemma \ref{lem5}. Hence using first the boundedness of the function $g_\beta$ and applying subsequently Lemma~\ref{lem3}, we get
\begin{multline*}
\Big|
\langle u_{j,\beta}, u_{l,\beta}\rangle_{L^2(\RR^2)}
-
\langle \varphi_{j,\beta}, \varphi_{l,\beta}\rangle_{L^2(\Pi(a))}
\Big|
\\
=
\Big|
\langle u_{j,\beta}, u_{l,\beta}\rangle_{L^2(\RR^2)}
-
\langle \varphi_{j,\beta}, \varphi_{l,\beta}\rangle_{L^2(\RR^2)}
\Big|\\
=
\bigg|
\iint_{\RR^2} \big(1-g_\beta(x)^2\big)\overline{u_{j,\beta}(x)}u_{l,\beta}(x)\,dx
\bigg|
=
\bigg|
\iint_{\RR^2\setminus S_\beta} \big(1-g_\beta(x)^2\big)\overline{u_{j,\beta}(x)}u_{l,\beta}(x)\,dx
\bigg|
\\
\le
C_1
\iint_{\RR^2\setminus S_\beta}
|\overline{u_{j,\beta}(x)}u_{l,\beta}(x)|\,dx
\le
C_2 \beta^4
\iint_{\RR^2\setminus S_\beta}
e^{-(\beta-\log\beta)d(x,\gamma)}
\,dx=\cO\big(\beta^{-2}\big)
\end{multline*}
with some constants $C_1,C_2>0$. As $\{u_{j,\beta}\}$ is an orthonormal system by assumption, we arrive at the relation \eqref{eq-ort}.
\end{proof}

\begin{lemma}
           \label{lem8}
In the limit $\beta\to+\infty$ one has
\[
\langle \nabla u_{j,\beta},\nabla u_{l,\beta}\rangle_{L^2(\Pi(a))}
-\beta \int_\gamma \overline{ u_{j,\beta}(s)}u_{l,\beta}(s)\,dS
=E_{j}(\beta)\delta_{jl} +\cO\big(\beta^{-1}\big).
\]
\end{lemma}
\begin{proof}
Note first that the relations
\[
\langle \nabla u_{j,\beta},\nabla u_{l,\beta}\rangle_{L^2(\RR^2)}-\beta \int_\gamma \overline{ u_{j,\beta}(s)}u_{l,\beta}(s)dS
=E_{j}(\beta)\delta_{jl}
\]
hold by assumption and that a certain neighborhood of $\gamma$ is included into $\Pi(a)$, hence it is sufficient to check the estimate
\[
\langle \nabla u_{j,\beta},\nabla u_{l,\beta}\rangle_{L^2(\RR^2\setminus\Pi(a))}=\cO(\beta^{-1}).
\]
As in the proof of Lemma \ref{lem6} we can check that the inclusion
\[
W\bigg(\dfrac{6\log\beta-C_1}{\beta}\bigg)\subset\Pi(a).
\]
holds for some $C_1>0$ and all sufficiently large $\beta$. Using then the estimate \eqref{eq-u2} and subsequently Lemma~\ref{lem3}, we get
\begin{multline*}
\Big|\langle \nabla u_{j,\beta},\nabla u_{l,\beta}\rangle_{L^2(\RR^2\setminus\Pi(a))}\Big|
\le
\iint_{\RR^2\setminus W\big(\frac{6\log\beta-C_1}{\beta}\big)} |\nabla u_{j,\beta}(x)|\cdot
|\nabla u_{l,\beta}(x)|\,dx\\
\le
C_2 \beta^6
\iint_{\RR^2\setminus W\big(\frac{6\log\beta-C_1}{\beta}\big)}
e^{-(\beta-\log\beta)d(x,\gamma)}\,dx
\le C_3\,\dfrac{\beta^6 }{\beta^7}=\cO\Big(
\dfrac{1}{\beta}\Big). \qedhere
\end{multline*}
\end{proof}

Our principal estimate concerns the question what happens if $u_{j,\beta}$ in the above formula is replaced by
the cut-off functions; our aim is to show that this makes the error term worse but only by a logarithmic factor.
\begin{lemma}
      \label{lempi}
In the limit $\beta\to+\infty$ one has
\[
\langle \nabla \varphi_{j,\beta},\nabla \varphi_{l,\beta}\rangle_{L^2(\Pi(a))}
-\beta \int_\gamma \overline{ \varphi_{j,\beta}(s)}\varphi_{l,\beta}(s)\,dS
=E_{j}(\beta)\delta_{jl} +\cO\Big(\dfrac{\log\beta}{\beta}\Big).
\]
\end{lemma}
\begin{proof}
Using $\varphi_{j,\beta}\big|_\gamma=u_{j,\beta}\big|_\gamma$
let us write the expression in question as
\begin{multline*}
\langle \nabla \varphi_{j,\beta},\nabla \varphi_{l,\beta}\rangle_{L^2\big(\Pi(a)\big)}-\beta \int_\gamma \overline{ \varphi_{j,\beta}(s)}\varphi_{l,\beta}(s)\,dS\\
=\langle \nabla u_{j,\beta},\nabla u_{l,\beta}\rangle_{L^2\big(\Pi(a)\big)}
-\beta \int_\gamma \overline{ u_{j,\beta}(s)}u_{l,\beta}(s)\,dS\\
+\iint_{\Pi(a)} \big(g_\beta(x)^2-1\big)\, \overline{\nabla u_{j,\beta}(x)}\cdot\nabla u_{l,\beta}(x)\,dx
+\iint_{\Pi(a)} |\nabla g_\beta(x)|^2 \,\overline{u_{j,\beta}(x)}u_{l,\beta}(x)\,dx\\
+\iint_{\Pi(a)} g_\beta(x) \overline{u_{j,\beta}(x)} \,\nabla g_\beta(x) \cdot \nabla u_{l,\beta}(x)\,dx
+\iint_{\Pi(a)} g_\beta(x) u_{l,\beta}(x) \,\overline{\nabla u_{j,\beta}(x)}\cdot\nabla g_\beta(x)\,dx.
\end{multline*}
The sum of the first two terms on the right-hand side has been already estimated in Lemma~\ref{lem8}, hence we just need to show that the sum of the last four terms on the right-hand side is of order $\cO(\beta^{-1}\log\beta)$. By definition of the function $g_\beta$ and Lemma \ref{lem6} there are constants $C_1,C_2>0$ such that
\begin{multline*}
I_1:=\bigg|\iint_{\Pi(a)} \big(g_\beta(x)^2-1\big) \overline{\nabla u_{j,\beta}(x)}\nabla u_{l,\beta}(x)dx\bigg|
=\bigg|\iint_{V(\beta)}
\big(g_\beta(x)^2-1\big) \overline{\nabla u_{j,\beta}(x)}\nabla u_{l,\beta}(x)dx\bigg|\\
\le C_1
\iint_{V(\beta)}
\Big|\overline{\nabla u_{j,\beta}(x)}\nabla u_{l,\beta}(x)\Big|dx
\le C_2 \big| V(\beta)\big|.
\end{multline*}
Furthermore, by definition of $V(\beta)$ we have
$V(\beta)=\Phi(U)$ with $U=\Big\{
(s,t)\in P(a):\, \rho_a(s,t)\le \beta^{-1}\Big\}$,
and since the measure $\big|U\big|$ is of order $\cO(\beta^{-1})$, we get also $\big|V(\beta)\big|=\cO(\beta^{-1})$, which in turn gives
$I_1=\cO(\beta^{-1})$.

Using next the inclusion $\supp\nabla g_\beta\subset\Theta(\beta)\subset V(\beta)$, Lemma \ref{lem6}
and after that Lemma \ref{lem4}, we have
\begin{multline*}
I_2:=\bigg|
\iint_{\Pi(a)} |\nabla g_\beta(x)|^2\, \overline{u_{j,\beta}(x)}u_{l,\beta}(x)\,dx
\bigg|\\
=
\bigg|
\iint_{\Theta(\beta)} |\nabla g_\beta(x)|^2\, \overline{u_{j,\beta}(x)}u_{l,\beta}(x)\,dx
\bigg|
\le
\dfrac{C_3}{\beta^2 }\iint_{\Theta(\beta)} \Big|\nabla g_\beta(x)\Big|^2\,dx
=\cO\Big(\dfrac{\log\beta}{\beta}\Big).
\end{multline*}
Using the same reasoning we infer that
\begin{multline*}
I_{j,l}:=\bigg|\iint_{\Pi(a)} g_\beta(x) \overline{u_{j,\beta}(x)}\, \nabla g_\beta(x) \nabla u_{l,\beta}(x)\,dx\bigg|\\
=
\bigg|\iint_{\Theta(\beta)} g_\beta(x) \overline{u_{j,\beta}(x)} \,\nabla g_\beta(x) \nabla u_{l,\beta}(x)\,dx\bigg|
\le
\dfrac{C_4}{\beta}
\iint_{\Theta(\beta)} \Big|\nabla g_\beta(x)\Big|\,dx =\cO\Big( \dfrac{1}{\beta}\Big).
\end{multline*}
Putting the estimates together we find
$I_1+I_2+I_{j,l}+I_{l,j}=\cO(\beta^{-1}\log\beta)$, which concludes the proof.
\end{proof}

Now we are in position to complete the proof of our main result.

\begin{proof}[Proof of Proposition \ref{prop2}]
Fix an integer $N\ge 1$. By the max-min principle one has
\[
\Lambda_N(\beta)=\max_{G\in S_N} \min_{0\ne f\in G} \dfrac{\displaystyle\iint_{\Pi(a)} |\nabla f|^2\,dx -\beta\displaystyle\int_\gamma|f|^2\,dS}{\|f\|^2_{L^2\big(\Pi(a)\big)}},
\]
where $S_N$ stands for the family of the subspaces of $H^1_0\big(\Pi(a)\big)$ the codimension of which in $L^2\big(\Pi(a)\big)$ equals $N-1$. In view of Lemma \ref{lemort},  the functions $\varphi_{j,\beta}$, $j=1,\dots,N$, are linearly independent in $L^2\big(\Pi(a)\big)$ for all sufficiently large $\beta$, hence each subspace $G\in S_N$ contains at least one linear combination $\varphi$ of the form
\[
\varphi=\sum_{j=1}^N b_j \varphi_{j,\beta}, \quad
b=(b_1,\dots,b_N)\in \CC^N, \quad \|b\|_{\CC^N}=1.
\]
Using once more Lemma \ref{lemort}, we find that
$\|\varphi\|^2_{L^2\big(\Pi(a)\big)}\ge1-C_1\beta^{-2}$
holds for large $\beta$ with a constant $C_1>0$. On the other hand, Lemma \ref{lempi} yields
\begin{multline}
\iint_{\Pi(a)} |\nabla \varphi|^2\,dx -\beta\int_\gamma|\varphi|^2\,dS
=\sum_{j,l=1}^N
\overline{b_j}b_l \bigg(
\langle \nabla \varphi_{j,\beta},\nabla \varphi_{l,\beta}\rangle_{L^2(\Pi(a))}
-\beta \int_\gamma \overline{ \varphi_{j,\beta}(s)}\varphi_{l,\beta}(s)dS
\bigg)\\
=\sum_{j,l=1}^N \overline{b_j}b_l \bigg(E_{j}(\beta)\delta_{jl} +\cO\Big(\dfrac{\log\beta}{\beta}\Big)\bigg)=
\sum_{j=1}^N E_j(\beta) |b_j|^2 +\cO\Big(\dfrac{\log\beta}{\beta}\Big)\\
\le E_N (\beta)+\cO\Big(\dfrac{\log\beta}{\beta}\Big).
\end{multline}
Using the above estimates, we conclude that there are $C_2,C_3>0$ such that
\begin{multline*}
\min_{0\ne f\in G} \dfrac{\displaystyle\iint_{\Pi(a)} |\nabla f|^2\,dx -\beta\displaystyle\int_\gamma|f|^2\,dS}{\|f\|^2_{L^2\big(\Pi(a)\big)}}
\,\le\,
\dfrac{\displaystyle\iint_{\Pi(a)} |\nabla \varphi|^2\,dx -\beta\displaystyle\int_\gamma|\varphi|^2\,dS}{\|\varphi\|^2_{L^2\big(\Pi(a)\big)}}\\
\le \dfrac{E_N (\beta)+C_2\dfrac{\log\beta}{\beta}}{1-C_1\beta^{-2}}\le E_N(\beta)+C_3 \dfrac{\log\beta}{\beta}.
\end{multline*}
What is important is that the constant $C_3$ can be chosen independent of the vector $b$ and hence independent of $G\in S_N$,
then we have automatically
\[
\Lambda_N(\beta)\le E_N(\beta)+ C_3\dfrac{\log\beta}{\beta}.
\]
Combining this with \eqref{eq-ela} we obtain  $\Lambda_N(\beta)-E_N(\beta)=\cO\big(\beta^{-1}\log\beta\big)$.
\end{proof}

\section{Acknowledgments}
The second named author thanks the Doppler Institute in Prague for the warm hospitality during the stay in May-June 2012.
The research was partially supported by ANR NOSEVOL and GDR DYNQUA, and by Czech Science Foundation within the project P203/11/0701.
The authors are thankful to the anonymous referee whose suggestions helped to shorten the presentation and to make it more transparent.


\begin{thebibliography}{999}

\bibitem{AS} M. Abramowitz, I. A. Stegun (eds.): Handbook of mathematical functions
with formulas, graphs, and mathematical tables. 10th printing (volume 55
of Applied Mathematics Series, US National Bureau of Standards, 1972).
Available online at \url{http://www.math.sfu.ca/~cbm/aands/}.

\bibitem{agr} M. S. Agranovich: \emph{Strongly elliptic second-order systems with boundary
vonditions on a nonclosed Lipschitz surface.}
Funct. Anal. Appl. {\bf 45}:1 (2011), 1--12.

\bibitem{BEKS} J. F. Brasche, P. Exner, Yu. A. Kuperin, P. \v Seba:
\emph{Schr\"odinger operators with singular ineractions.}
J. Math. Anal. Appl. {\bf 184} (1994), 112--139.

\bibitem{DH} M. Dauge, B. Helffer: \emph{Eigenvalues variation I. Neumann problem for Sturm-Liouville operators.}
J. Differential Eqs {\bf 104} (1993), 243--262.

\bibitem{E08} P.~Exner: {\em Leaky quantum graphs: a review}, Proceedings of the Isaac Newton Institute programme ``Analysis on Graphs and Applications'', AMS ``Proceedings of Symposia in Pure Mathematics'' Series, vol.~77, Providence, R.I., 2008; pp.~523--564.

\bibitem{EI01} P. Exner, T. Ichinose: \emph{Geometrically induced spectrum in curved leaky wires.}
J. Phys. A: Math. Theor. {\bf 34} (2001) 1439--1450.

\bibitem{ET} P. Exner, M. Tater: \emph{Spectra of soft ring graphs.}
Waves Random Complex Media {\bf 14}:1 (2004) S47--S60.


\bibitem{EY} P. Exner, K. Yoshitomi: \emph{Asymptotics of eigenvalues of the Schr\"odinger operator
with a strong $\delta$-interaction on a loop.} J. Geom. Phys. {\bf 41} (2002), 344--358.

\bibitem{grubb}
G. Grubb: \emph{The mixed boundary value problem, Krein resolvent formulas and spectral asymptotic estimates.}
 J. Math. Anal. Appl. {\bf 382} (2011), 339--363.

\bibitem{KV} S. Kondej, I. Veseli\'c: \emph{Lower bounds on the lowest spectral gap of singular potential Hamiltonians.}
Annales Henri Poincar\'e {\bf 8} (2006) 511--552.
 
\bibitem{lp}
M. Levitin, L. Parnovski:
\emph{On the principal eigenvalue of a Robin problem with a large parameter.}
Math. Nachrichten {\bf 281} (2008), 272--281.

\bibitem{pos} A. Posilicano: \emph{Boundary triples and Weyl functions for singular perturbations of selfadjoint
operator}, Meth. Funct. Anal. Topol. {\bf 10} (2004), 57--63.

\end{thebibliography}
\end{document}